\documentclass{article}
\usepackage[pdftex]{graphicx}
\usepackage[utf8]{inputenc}
\usepackage[paper=letterpaper, margin=1in]{geometry}

\usepackage{amsthm}
\usepackage{subcaption}
\usepackage{amsmath}
\usepackage{hyperref}

\usepackage{authblk}

\newtheorem{theorem}{Theorem}
\newtheorem{lemma}[theorem]{Lemma}

\newtheorem{claim}[theorem]{Claim}
\theoremstyle{definition}

\newcommand{\pa}[1]{\left( #1 \right)}
\newcommand{\DeltaD}{\Delta_{\bar{d}}}
\newcommand{\D}{\bar{D}}
\newcommand{\diam}{\textnormal{diam}}
\newcommand{\floor}[1]{\lfloor #1 \rfloor}

\title{On the Role of Hypocrisy in Escaping the Tragedy of the Commons}

\author[1]{Amos Korman}
\author[2]{Robin Vacus}
\affil[1]{\small Corresponding author. UMI FILOFOCS, CNRS, UP7, TAU, HUJI, WIS, International Joint Research Unit.  E-mail: {\tt amos.korman@irif.fr}}
\affil[2]{\small IRIF, CNRS and University of Paris, France.  E-mail: {\tt rvacus@irif.fr}}

\begin{document}

%\begin{nolinenumbers}
\date{}
\maketitle

\begin{abstract}
{\bf
We study the emergence of cooperation in large spatial public goods games. Without employing severe social-pressure against ``defectors'', or alternatively, significantly rewarding ``cooperators'', theoretical models typically predict a system collapse in a way that is reminiscent of the ``tragedy-of-the-commons'' metaphor. Drawing on a dynamic network model, this paper demonstrates how cooperation can emerge when the social-pressure is mild. This is achieved with the aid of an additional behavior called ``hypocritical'', which appears to be cooperative from the external observer's perspective but in fact hardly contributes to the social-welfare. Our model assumes that social-pressure is induced over both defectors and hypocritical players, but the extent of which may differ. Our main result indicates that the emergence of cooperation highly depends on the extent of social-pressure applied against hypocritical players. Setting it to be at some intermediate range below the one employed against defectors allows a system composed almost exclusively of defectors to transform into a fully cooperative one quickly. Conversely, when the social-pressure against hypocritical players is either too low or too high, the system remains locked in a degenerate configuration.
}
\end{abstract}

\section*{Introduction}
The ``tragedy-of-the-commons'' metaphor, popularized by Hardin in 1968~\cite{hardin}, aims to capture situations in public goods systems where self-interested individuals behave contrary to the common good by depleting or spoiling the shared resource. 
In the 21st century, this metaphor finds relevance in several of our global environmental challenges \cite{houghton2001climate,ostrom1990governing}, where the shared resource can be considered, depending on the context, as an aspect of the ecosystem. For example, excessive beef consumption by a substantial number of individuals induces vast livestock production that degrades air and water quality and causes a considerable increase in greenhouse gas emissions \cite{eshel2014land}. Conversely, our environment would significantly benefit if a large portion of individuals in the population would self-restraint the amount of beef they consume. Therefore, improving our understanding of the emergence of cooperation in public goods systems goes beyond the purely theoretical interest and may prove to be of practical importance. 

Theoretical studies on the emergence of cooperation typically assume that players act according to few stereotyped behaviors, the most common being ``defector'', and ``cooperator'' \cite{ohtsuki2006simple,nowak2006five,taylor2007evolution,allen2019evolutionary,allen2017evolutionary}.
A cooperator pays an energetic cost to produce a benefit $b$ for others, whereas a defector does not contribute anything but also does not pay any energetic cost. 
In recent years, significant attention has been devoted to study the impact of the populations' structure on the emergence of cooperation \cite{perc2013evolutionary,ohtsuki2006simple,allen2014games,debarre2014social}. 
These works assume that players are organized over a fixed network, 
with the vertices representing the players and the edges representing reciprocal relations between neighbors. Naturally, the dynamics of the system strongly depend on the mutual relations between neighboring players. 

For example, several of the works on cooperation in structured populations assume that the  benefit $b$ produced by a cooperative player is shared equally by its neighbors. For such a model, Ohtsuki et al.~showed that cooperation emerges when the ratio between the benefit per per neighbor and the cost of producing it exceeds a certain threshold \cite{ohtsuki2006simple}. However, large public goods games, especially those on the scale that affects the environment, exhibit a very different framework of reciprocity \cite{cinyabuguma2006can,milinski2002reputation,rand2009positive,rege2004impact,yamagishi1986provision}. Rather than being shared by immediate neighbors,
the benefit $b$ is shared by all individuals, practically making the marginal per-capita return gain (MPCR) negligible compared to the cost of cooperating. This violates the condition for the evolution of cooperation based on reciprocity \cite{ohtsuki2006simple,nowak2006five,taylor2007evolution,allen2019evolutionary,allen2017evolutionary} suggesting that cooperation in large public goods games might be difficult to achieve without considering other factors, such as rewards or punishments. 

It is well-known that global cooperation can emerge when players severely punish their neighboring defectors (or, alternatively, significantly reward their cooperating neighbors) \cite{sigmund2001reward,nowak2006five,rankin2007tragedy,axelord1984evolution,milinski2002reputation,fehr2000cooperation}. However, inducing severe punishments on others may be  costly, and hence reaching high levels of social-pressure is by itself a non-trivial problem, often referred to in the literature as the {\em second-order free riders} problem
\cite{boyd1992punishment,boyd2003evolution,cinyabuguma2006can,eldakar2008selfishness,fowler2005altruistic,heckathorn1989collective,heckathorn1996dynamics,helbing2010evolutionary,helbing2010punish,panchanathan2004indirect,yamagishi1986provision}.
A crucial parameter in the second-order problem is the cost of punishing, which may be correlated to the extent of punishment \cite{rockenbach2006efficient}. Clearly, when the cost exceeds a certain threshold, people would avoid punishing non-cooperators. However, when the cost is low, other factors, such as reputation considerations, can subsume the cost, ultimately making punishing beneficial \cite{jordan2016third,wedekind2000cooperation,zahavi1995altruism}. It is therefore of interest to study the emergence of cooperation in the presence of moderate punishments or mild social-pressure. 

Specifically, we are interested in a regime of social-pressure that is high enough to maintain an already cooperative system, but is insufficient to transform a system that initially includes a large number of defectors into a cooperative one. To illustrate this, let us consider the context of recycling and an imaginary person named Joe. When almost all of Joe's neighbors are recycling (i.e., cooperating), the social-pressure cost they induce on him can accumulate to overshadow the burden cost of recycling and incentivize him to also recycle. Conversely, when almost all of Joe's neighbors are not recycling (i.e., defecting), the burden of recycling may exceed the overall social-pressure, effectively driving Joe to defect. This raises a natural question: 
\begin{center}
    {\em How can a  system that utilizes mild social-pressure\\  escape the tragedy-of-the-commons when it is\\ initially composed mostly of defectors?}
    \end{center}
\vspace{2mm} 

The aforementioned recycling abstraction  includes two extreme behaviors: defector and cooperative. Another type of  generic behavior is {\em hypocritical} \cite{trivers1971evolution,centola2005emperor,heckathorn1989collective,heckathorn1996dynamics,helbing2010evolutionary,eldakar2008selfishness},
which was also experimentally studied in \cite{falk2005driving,shinada2004false}. In our interpretation, a hypocritical individual pretends to be cooperative in order to reduce the social-pressure that it might experience as a defector, and, at the same time, avoids the high energetic cost incurred by a cooperator. To pretend to be a cooperator, a hypocritical individual must  invest a small amount of energy in contributing to the social welfare, as well as mimic the behavior of cooperators towards their peers. This means that such players, similarly to cooperators, also induce mild  social-pressure. 
It was previously suggested that hypocritical behavior can incentivize global cooperation 
\cite{helbing2010evolutionary,helbing2010punish}. However, in these works, similarly to many other papers on the emergence or evolution of cooperation based on reciprocity  \cite{ohtsuki2006simple,nowak2006five,taylor2007evolution,allen2019evolutionary,allen2017evolutionary}, the dynamics heavily relies on the assumption that players gain substantially from the presence of nearby cooperators. As mentioned, this assumption is hardly justifiable in large-scale public goods scenarios such as the ones we consider.  

\section*{Results}

We consider public goods games played iteratively over a fixed connected network. 
The vertices of the network represent the players and the edges represent neighboring connections \cite{perc2013evolutionary,ohtsuki2006simple,allen2014games,debarre2014social}. 
The dynamics evolves over discrete rounds. In each round, each player chooses a behavior that minimizes its cost, where the player's cost is affected by its own behavior and the behaviors of its neighbors.

Our main model includes three behavior types, namely, defector, hypocritical, and cooperator, in which those who hardly contribute to the social welfare, i.e., defector and hypocritical players, face the risk of being caught and punished by their non-defector neighbors. 
The level of risk together with the extent of punishment is captured by a notion that we call {\em ``social-pressure''}.
The main result is that adjusting the level of social-pressure employed against hypocritical players compared to the one employed against defectors can have a dramatic impact on the dynamics of the system. Specifically, letting the former level of social-pressure be within a certain range below the latter level, allows the system to quickly transform from being composed almost exclusively of defectors to being fully cooperative. Conversely, setting the level to be either too low or too high locks the system in a degenerate configuration.

As mentioned, our main model assumes that non-defector players induce mild social-pressure on their defector neighbors. This implicitly assumes that inducing the corresponding  social-pressure is  beneficial (e.g., allows for a social-upgrade), although other explanations have also been proposed \cite{fehr2000cooperation}. To remove this implicit assumption we also consider a generalized model, called the {\em two-order model}, which includes costly punishments. Consistent with previous work on the second-order problem, e.g., 
\cite{boyd1992punishment,eldakar2008selfishness,fowler2005altruistic,heckathorn1996dynamics,helbing2010evolutionary,helbing2010punish},
this model distinguishes between first-order cooperation, that corresponds to  actions that directly contribute to the social welfare, and second-order cooperation, that corresponds to applying (costly) social-pressure, or punishments, on others.
Similarly to the main model, the level of punishment employed against first-order defectors may differ from that employed against second-order defectors. 
We identify a simple criteria for the emergence of cooperation: For networks with minimal degree $\Delta$, cooperation emerges when two conditions hold. The first condition states that the cost $\alpha_2$ of employing  punishments against second-order defectors should be smaller than the corresponding punishment $\beta_2$ itself, i.e., $\alpha_2<\beta_2$. The second condition states that the cost $\alpha_ 1$ of employing punishments against first-order defectors should be smaller than the corresponding punishment $\beta_1$ times the minimal number of neighbors, i.e., $\alpha_ 1<\beta_1\cdot\Delta$. The second condition is also a necessary condition for the emergence of cooperation in the two-order model.

\subsection*{The main model}
The model considers two extreme behaviors, namely, {\em cooperative} $(c)$ and {\em defector} $(d)$, and an additional intermediate behavior, called {\em hypocritical} $(h)$. 
The system starts in a configuration in which almost all players, e.g., $99\%$, are defectors (see Methods). 
Execution proceeds in discrete rounds. The cost of a player depends on its own behavior and on the behavior of its neighbors.  All costs are evaluated at the beginning of each round, and then, before the next round starts, each player chooses 
a behavior that minimizes its cost (breaking ties randomly), given the current behavior of its neighbors. 
In contrast to many previous works on cooperation in networks \cite{ohtsuki2006simple,nowak2006five,taylor2007evolution,allen2019evolutionary,allen2017evolutionary}, we assume that benefits from  altruistic acts are negligible (i.e., the MPCR is zero), so that a player does not gain anything when others cooperate.
 
The cost of a player $u$ with a behavior type $i\in\{d,h,c\}$ is composed of two components: the {\em energetic cost} $E_i$ associated with the contribution to the social welfare, and the {\em social-pressure cost} $S_i(u)$ it faces, that is:
\[{\cal{C}}_i(u)=E_i+S_i(u).\]
We assume that the energetic cost of a defector is 0, and the energetic cost of a cooperator is $1$,  where the value of $1$ is chosen for normalization:
\[E_{d}=0~~ \mbox{and}~~ E_{c}=1.\]
A hypocritical player produces the minimal social welfare required to pretend to be cooperative. 
Hence, we assume that
\[
0<E_h<1,
\]
thinking of $E_h$ as closer to 0 than to 1. 

As mentioned above, we focus on relatively mild social-pressure induced by cooperative players, aiming to improve their social status. Since hypocritical players aim to appear similar to cooperators from the perspective of an external observer, we assume that they too induce social-pressure on their neighbors. Defectors, on the other hand, do not induce any social-pressure since such an enhancement of the social status is not justified for them. In principle, cooperators and hypocritical players might induce different levels of social-pressure, yet, for the sake of simplicity, we assume that they induce the same extent of social-pressure. This assumption is further justified by the fact that a player $u$ cannot distinguish its hypocritical neighbors from its cooperative neighbors, hence, $u$'s calculation of the social-pressure is evaluated assuming all of its non-defector neighbors are cooperators. 

Formally, we assume that the possible social-upgrade gain associated with cooperators or hypocritical players as a result of applying social-pressure is already taken into account when calculating the energetic costs $E_c$ and $E_h$. Since we assume that this gain is small, it hardly perturbs the cost, keeping the energy consumption as the dominant component. 

Implicitly, we think of  the social-pressure cost incurred by a player $u$
as the product of two factors: (1) the risk of being caught, which is assumed to be proportional to the number of $u$'s neighbors inducing social-pressure, and (2) a fixed penalty paid when caught, which depends on $u$'s behavior. The product of the risk and penalty represents the expected punishment in the next round, if behaviors remain the same. 

Cooperators are assumed to pay zero penalty, and are hence effectively immune to social-pressure:
\[S_{c}(u)=0.\]
Conversely, the social-pressure induced over defectors and hypocritical players is non-zero.
For a given round, let $\DeltaD(u)$ denote the number of neighbors of $u$ which are non-defectors at that round. The social-pressure cost induced over a defector, and respectively, a hypocritical, player $u$ is: 
\[S_{d}(u)=\rho_{d}\cdot \DeltaD(u),~~~\mbox{respectively, }~~~S_{h}(u)=\rho_{h}\cdot \DeltaD(u),\]
where $\rho_{d} >0$, respectively $\rho_{h} >0$, represents the social-pressure  induced over a defector, respectively a hypocritical, from one neighboring non-defector. We focus on the regime where $\rho_h < \rho_d$, since otherwise, becoming a defector is always more beneficial than becoming a hypocritical.

To sum up, at a given round, the total cost incurred by a player $u$ is:
\begin{equation*}
    {\cal{C}}(u) = \begin{cases} 1 & \mbox{if $u$ is a cooperator,} \\
    \rho_{d} \cdot \DeltaD(u) & \mbox{if $u$ is a defector,} \\
    E_h + \rho_{h} \cdot \DeltaD(u) & \mbox{if $u$ is hypocritical.} \end{cases}
\end{equation*}

Before stating our main result, we recall few standard definitions in graph-theory \cite{diestel2005graph}. The {\em diameter} of a network $G$, denoted $ \diam(G)$, is the maximal distance between any pair of players (see Methods). A network is {\em $\Delta$-regular}, if every player has precisely $\Delta$ neighbors. Theorem~\ref{thm:main} below assumes that the underlying network is $\Delta$-regular. 
However, this theorem can be generalized to arbitrary networks with minimal degree $\Delta$ (see SI, Theorem \ref{thm:main-SI}).
\begin{theorem}\label{thm:main}
Consider a $\Delta$-regular network $G$ with $n$ players. Assume that 
\begin{equation} \label{eq1}
    ({1-E_h})/{\Delta}<\rho_{h}<\rho_{d}-E_h.
\end{equation}
Then, with probability at least $1-\frac{1}{c^n}$, for some constant $c>1$, in at most $3 \cdot \diam(G)+1$ rounds, the system will be in a configuration in which all players are cooperative, and will remain in this configuration forever.
\end{theorem}

\begin{figure}[!ht]
    \centering
    \includegraphics[width=0.5\linewidth]{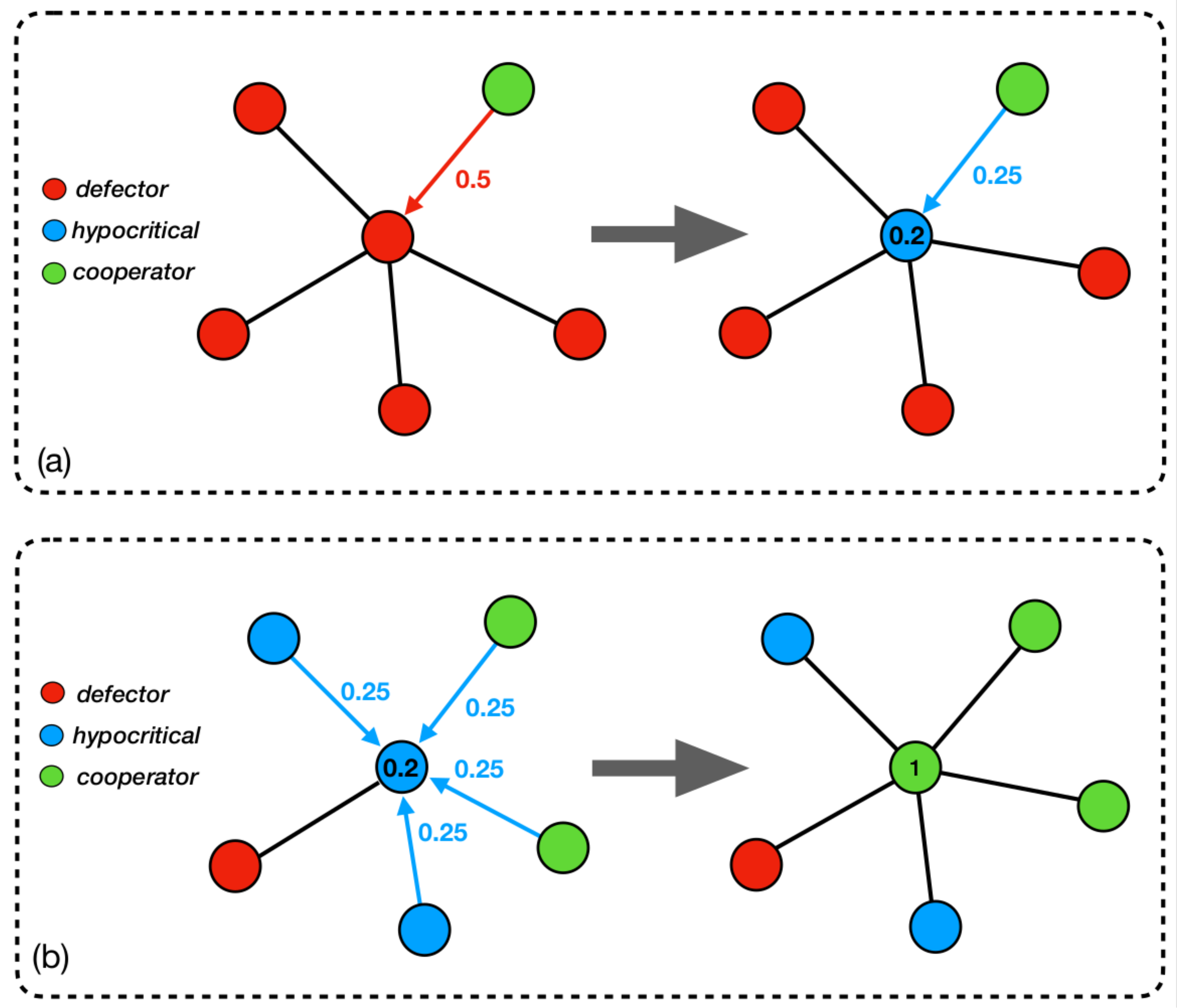}
        \caption{{\bf The two stages of the  dynamics.} 
        The direction of the red and blue arrows indicates the direction of the social-pressure applied on the player occupying the central vertex.  Cooperative players pay an energetic cost of $E_c=1$ and are immune to social-pressure. A defector player pays a social-pressure cost of $\rho_{d}=0.5$ per non-defector neighbor. A hypocritical player pays social-pressure cost of $\rho_{h}=0.25$ per non-defector neighbor, and an energetic cost of $E_h=0.2$.
        (a) First stage: defectors become hypocritical players. A defector player (central vertex on the left) has one non-defector neighbor (in this case, a cooperator), implying that its social-pressure cost is $\rho_{d}=0.5$. Therefore, that player would prefer to be hypocritical (right), paying only $0.25+0.2=0.45$. (b) Second stage: hypocritical players become cooperators. Here, a hypocritical player (central vertex on the left) is surrounded by four non-defector neighbors. In this case, the social-pressure accumulates to favor cooperation (right).}\label{fig:double-dynamics}
\end{figure}

The formal proof of Theorem~\ref{thm:main} appears in the SI,  Section~\ref{secSI:proofs}.
Intuitively, the main idea behind it is as follows. When the extent of social-pressure against hypocritical players is moderate, that is, when $\rho_h$ satisfies Eq.~\eqref{eq1}, the transition process can be divided into two stages. At the first stage, since the punishments of hypocritical players are sufficiently lower than those of defectors, specifically, $\rho_{h}<\rho_{d}-E_h$, or equivalently $\rho_{h}+E_h<\rho_{d}$, the presence of at least one neighboring non-defector $u$ makes a hypocritical player pay less than a defector. In this case, $u$'s neighbors  would become non-defectors at the next round (Figure~\ref{fig:double-dynamics}a). Although this does not necessarily imply that $u$ itself remains a non-defector in the next round, it is nevertheless possible to show that the proportion of hypocritical players gradually increases on the expense of defectors. 
Note that at this point, the social welfare may still remain low, since hypocritical players hardly contribute to it. However, the abundance of non-defectors increases the overall social-pressure in the system. In particular, since the social-pressure on hypocritical players is also not too mild, specifically $(1-E_h)/{\Delta}<\rho_{h}$, or equivalently $1<\rho_{h}\Delta+E_h$, the presence of many neighboring non-defectors can magnify it up to the point that the total cost incurred by a hypocritical player surpasses the energetic cost of being a cooperator (Figure~\ref{fig:double-dynamics}b). At this second stage, cooperators prevail over both defectors and hypocritical players, and so the system converges to a cooperative configuration. 

Conversely, severely punishing hypocritical players diminishes the prevalence of such players,  
preventing the system from escaping the initial degenerate configuration. Contrariwise, incurring too mild social-pressure towards hypocritical players would prevent the second stage of the dynamics. In particular, 
if $\rho_h < (1-E_h)/\Delta$, or equivalently, if $E_h + \rho_h \Delta < 1$, then a player would always prefer to be hypocritical over being cooperative  (even when all its neighbors induce social-pressure). In this case, the system would remain degenerative since the population would consist of a combination of defectors and hypocritical players.

To illustrate the dynamics
we conducted simulations  over  
two types of networks: A two-dimensional torus grid, and random 10-regular networks. 
Figures~\ref{fig:grid_evolution} (grid) and~\ref{fig:regular_evolution} (random $10$-regular networks) show how the population evolves over time, where the parameters taken satisfy Eq.~\eqref{eq1}. The role of hypocritical behavior as a transitory state, essential to achieving cooperation, is well illustrated by the initial peak of hypocritical players, preceding the rise of cooperative players. Moreover, if hypocritical behavior is disabled (see Methods), then the system is unable to escape the defective state (insets). 

Figures~\ref{fig:grid_eh_rhoh} (grid) and~\ref{fig:regular_eh_rhoh} (random $10$-regular networks) show the steady-state configuration, when hypocritical players experience different levels of energetic cost ($E_h$) and social-pressure ($\rho_h$). 
 For small values of $\rho_h$ and $E_h$, hypocritical behavior is, unsurprisingly, dominant: punishments deter  defectors, but are  insufficient to incentivize cooperation.  For moderate values of $E_h$, this phenomenon changes when $\rho_h$ enters the range specified in Theorem~\ref{thm:main}. Then, when $\rho_h$ increases further, the system remain defective.

As it turns out, the convergences we see in Figures~\ref{fig:grid_eh_rhoh} and~\ref{fig:regular_eh_rhoh} are very strong, in the sense that almost all players have the same behavior at steady-state. This unrealistic outcome is a consequence of several simplifying assumptions, including the fact that players behave in a fully greedy fashion while having perfect knowledge regarding their costs. Indeed, we also simulated a more noisy variant of our model, in which each player chooses the behavior that minimizes its cost with probability $0.95$, and otherwise chooses a behavior uniformly at random. This relaxed model yields more mixed populations at steady-state (Figures~\ref{fig:grid_eh_rhoh_EG} and~\ref{fig:regular_eh_rhoh_EG}). Observe that the necessity of the condition $\rho_h > (1-E_h)/\Delta$ to the emergence of cooperation is still respected. However, the other condition, namely, $\rho_h < \rho_d - E_h$ appears to be more sensitive to randomness. Indeed, for random $\Delta$-regular graphs, 
cooperation emerges also for larger values of $\rho_h$.

\begin{figure*}[htbp]
\begin{subfigure}[t]{.38\textwidth}
  \centering
  \includegraphics[width=1.\linewidth]{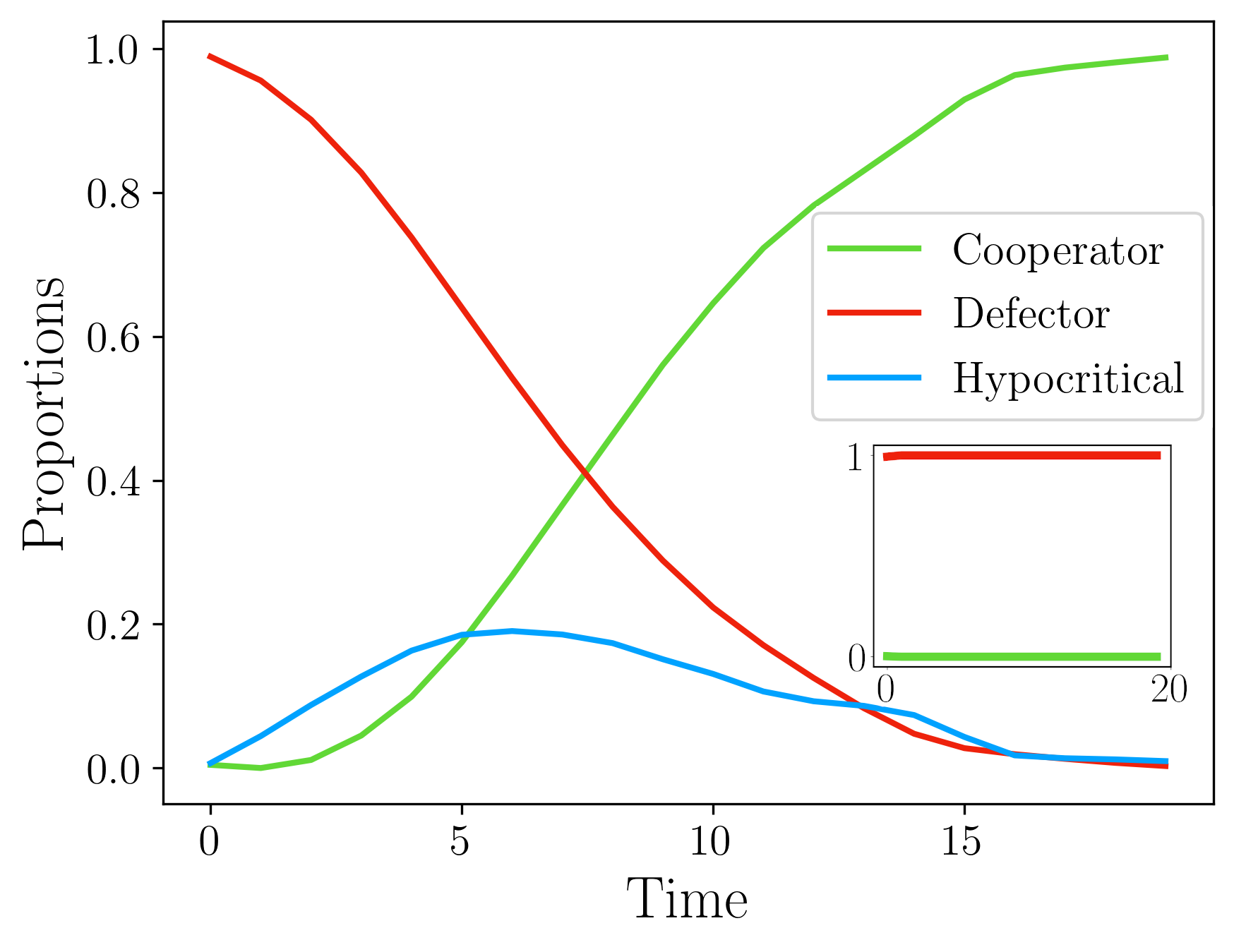}  
  \caption{Time evolution on a $50\times 50$  grid}
  \label{fig:grid_evolution}
\end{subfigure}
\hfill
\begin{subfigure}[t]{.38\textwidth}
  \centering
  \includegraphics[width=1.\linewidth]{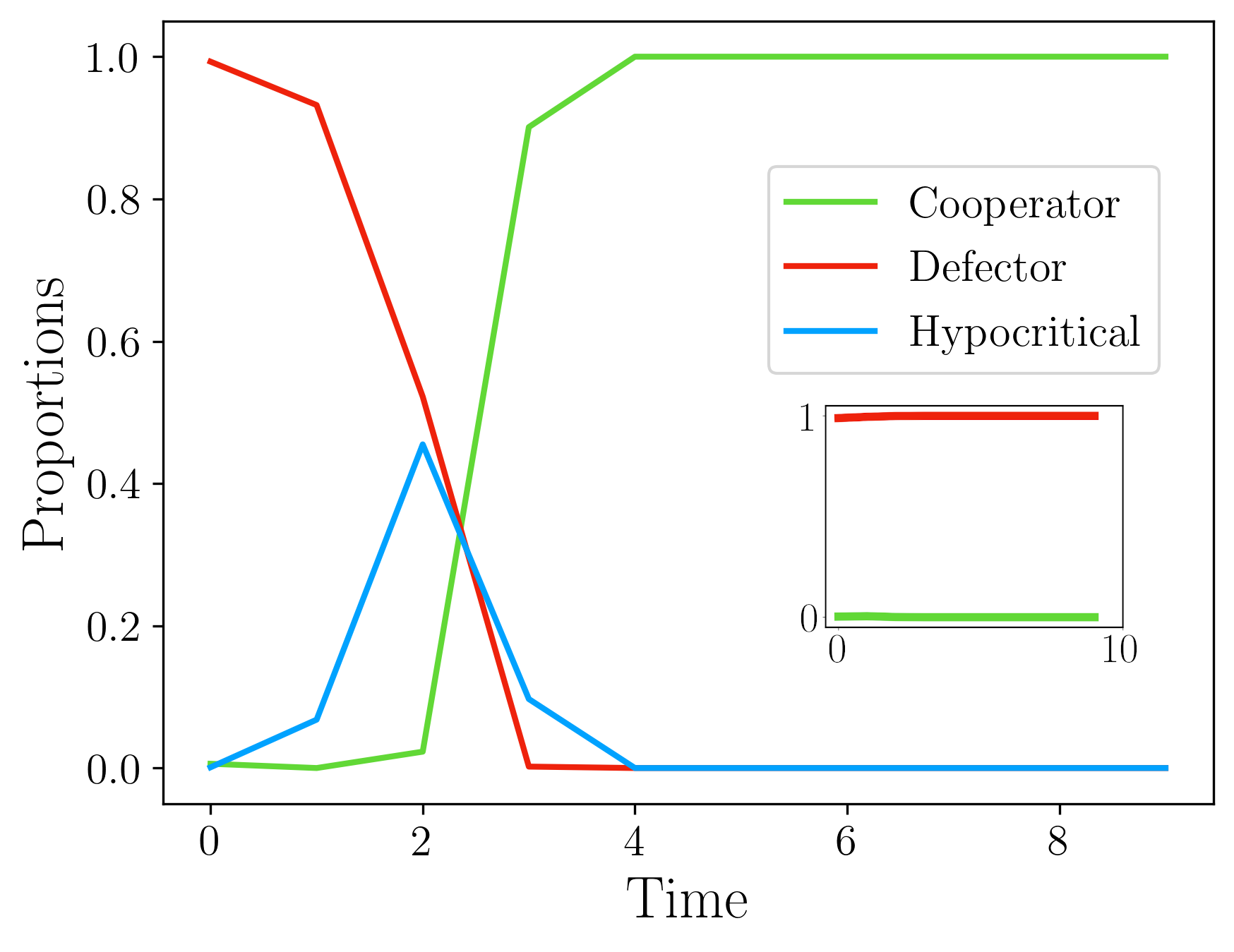}  
  \caption{Time evolution on a $10$-regular network}
  \label{fig:regular_evolution}
\end{subfigure}

\begin{subfigure}[t]{.45\textwidth}
  \centering
  \includegraphics[width=1.\linewidth]{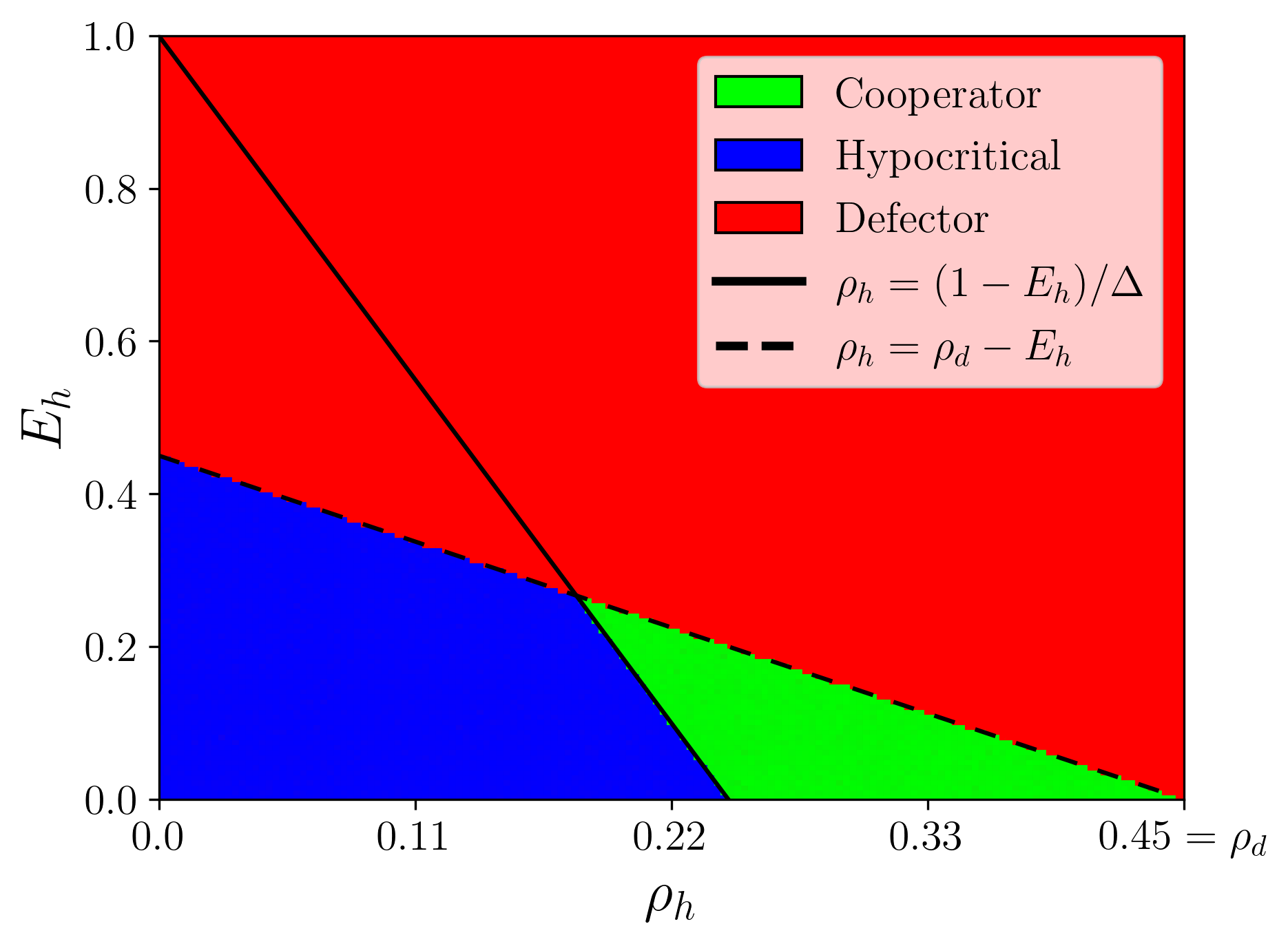}  
  \caption{Cooperation level on a $50\times 50$ grid}
  \label{fig:grid_eh_rhoh}
\end{subfigure}
\hfill
\begin{subfigure}[t]{.45\textwidth}
  \centering
  \includegraphics[width=1.\linewidth]{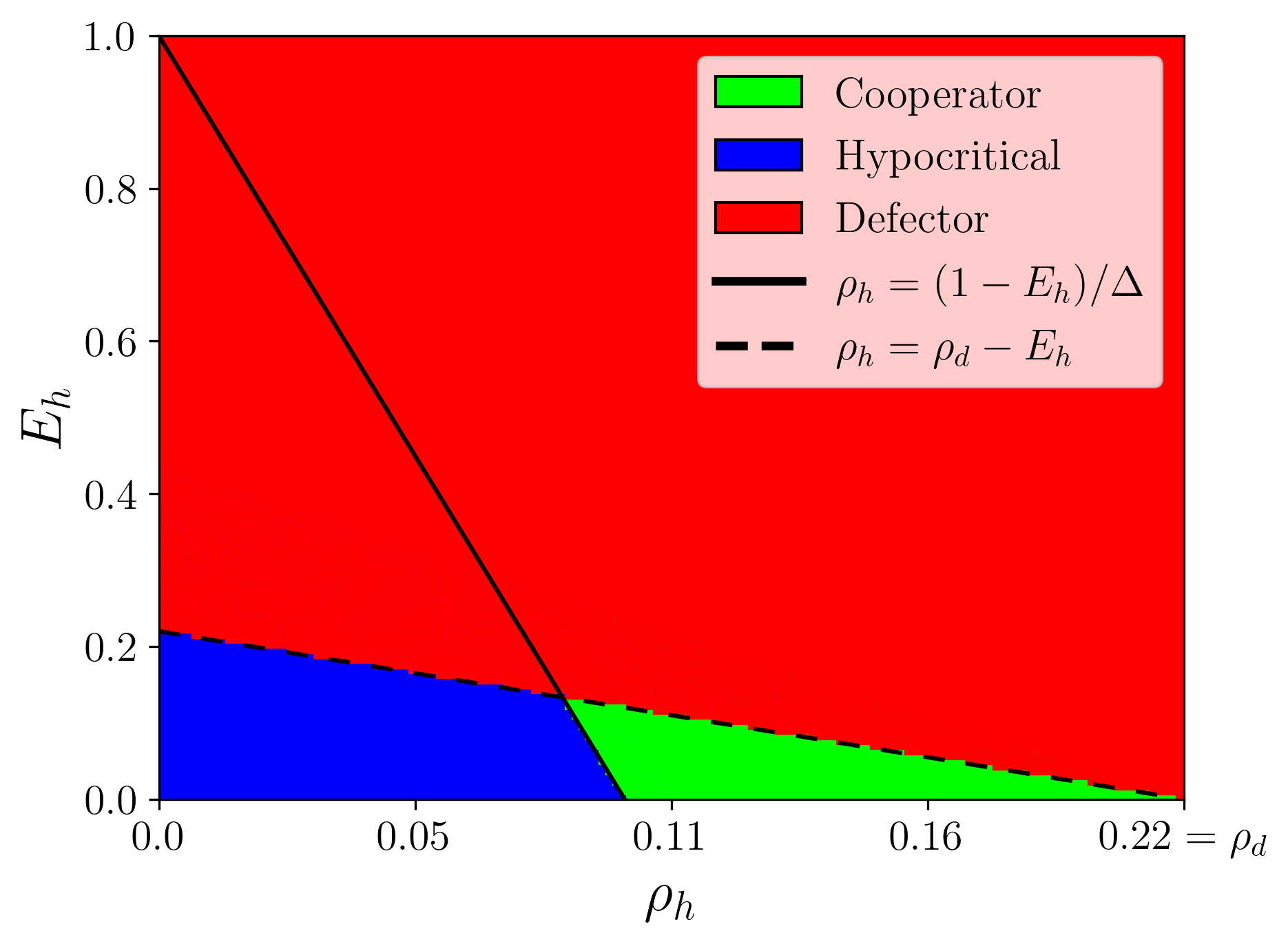}  
  \caption{Cooperation level on 10-regular networks}
  \label{fig:regular_eh_rhoh}
\end{subfigure}

\begin{subfigure}[t]{.45\textwidth}
  \centering
  \includegraphics[width=1.\linewidth]{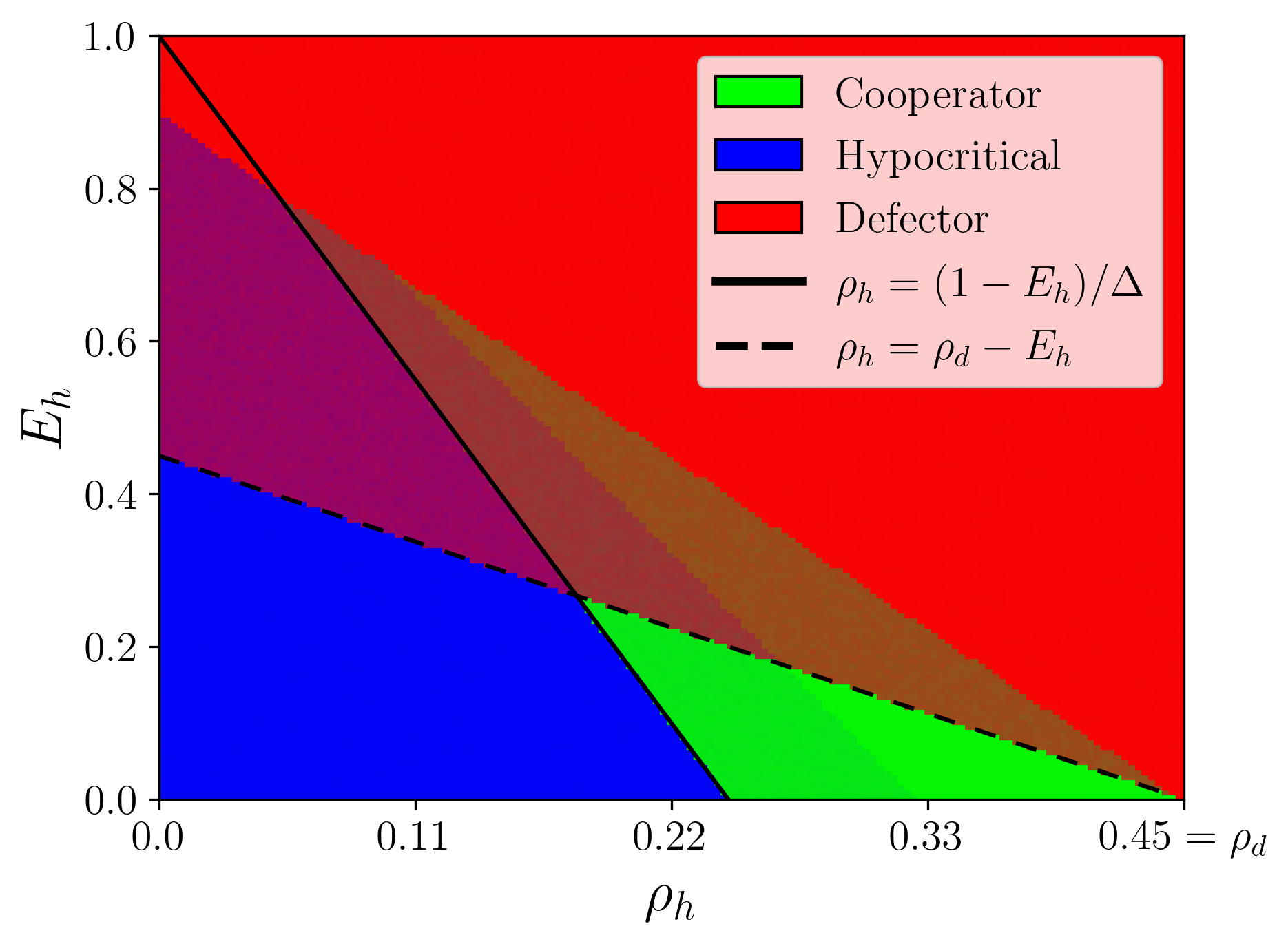}  
  \caption{Relaxed model on a $50\times 50$ grid}
  \label{fig:grid_eh_rhoh_EG}
\end{subfigure}
\hfill
\begin{subfigure}[t]{.45\textwidth}
  \centering
  \includegraphics[width=1.\linewidth]{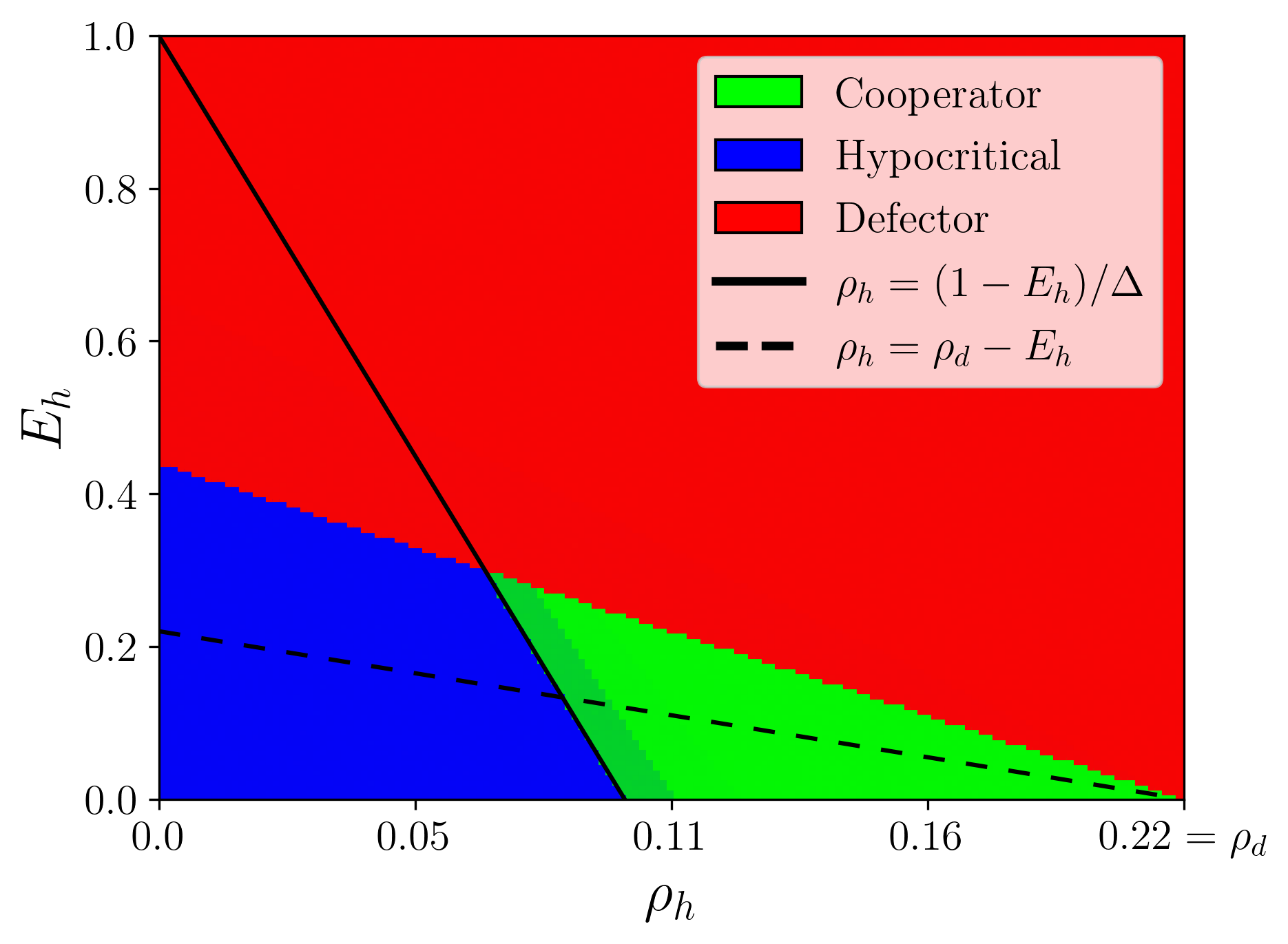}  
  \caption{Relaxed model on 10-regular networks}
  \label{fig:regular_eh_rhoh_EG}
\end{subfigure}

\caption{{\bf Emergence of cooperation in grids and random 10-regular networks.} Figures (a), (c) and (e) correspond to a $50\times 50$ grid network, and Figures (b), (d) and (f) correspond to random 10-regular networks with $1000$ vertices. All simulations start with a configuration in which $99\%$ of players are defectors.
Figures (a) and (b) show how the population evolves over time (number of rounds).  The chosen parameters satisfy the assumption in Eq.~\eqref{eq1} (see Methods). The insets show the population's evolution when hypocritical behavior is not available to the agents. 
Figures (c) and (d) depict the steady-state behavior, for different levels of $\rho_h$, which is the parameter quantifying the social-pressure towards hypocrisy. 
Figures (e) and (f) are similar to Figures (c) and (d), except that we relaxed the greediness assumption in the decision making process, allowing for some ``irrationality''. See Methods for more details.}
\label{fig:evolution}
\end{figure*}

\subsection*{A Generalized Model with Costly Punishments}
\newcommand{\spc}{\beta_1}
\newcommand{\spp}{\beta_2}

We next describe a different, more general model, termed the {\em two-order model}, that includes costly punishments. We then show how the second-order problem is solved in this model for a certain regime of parameters. 

In the two-order model, similarly to the main model, players are organized over a connected network $G$. A {\em behavior} for Player~$u$ is defined as a couple of indicator functions $(\chi_1(u),\chi_2(u))$, with the convention that $\chi_1(u) = 1$ if~$u$  cooperates on the first-order (and $0$ if it defects), and $\chi_2(u) = 1$ if~$u$  cooperates on the second-order (and $0$ if it defects).

The cost incurred by a player is  divided into two components. We denote by $\alpha_1 > 0$ the cost associated with first-order cooperation (this is analogues to the energetic cost in the main model), while $\alpha_2 > 0$ refers to the cost of second-order cooperation, that is, the cost of incurring punishments.
A player~$u$ such that~$\chi_2(u) = 1$ induces a {\em social-pressure cost} on each of its neighbors, whenever these fail to cooperate, at any order. Similarly to the main model, the extent of this social-pressure may differ depending on whether it is applied against first-order defectors or second-order defectors. Specifically, we denote by $\beta_1$ the social-pressure cost paid by 
a first-order defector, and by $\beta_2$ the social-pressure cost paid by a second-order defector (fully defecting individuals pay both).
Formally, denoting by $\Delta_2(u)$ the number of neighbors of $u$ which are cooperating on the second-order, that is, $ \Delta_2(u) = |\{\mbox{$v$ is a neighbor of $u$}, \chi_2(v) = 1 \}| $, the total cost paid by~$u$ equals:
\begin{equation}
    {\cal{C}}(u) = \chi_1(u) \alpha_1 + \chi_2(u) \alpha_2 + (1-\chi_1(u)) \Delta_2(u) \spc + (1-\chi_2(u)) \Delta_2(u) \spp.
\end{equation}
Let us name each of the four behaviors, and recap their cost:
\begin{itemize}
    \item {\em a cooperator} ($\chi_1(u) = 1 , \chi_2(u) = 1$) pays $\alpha_1 + \alpha_2$,
    \item {\em a defector} ($\chi_1(u) = 0 , \chi_2(u) = 0$) pays $ \Delta_2(u) (\spc + \spp)$,
    \item {\em a hypocritical} ($\chi_1(u) = 0 , \chi_2(u) = 1$) pays $ \alpha_2 + \Delta_2(u) \spc$,
    \item {\em a private cooperator} ($\chi_1(u) = 1 , \chi_2(u) = 0$) pays $ \alpha_1 + \Delta_2(u) \spp$.
\end{itemize}
Similarly to the main model, the system starts in a configuration in which almost all players, e.g., $99\%$, are defectors (see Methods). The execution proceeds in discrete synchronous rounds. The costs of each player are evaluated at the beginning of each round, and then, before the next round starts, each player chooses a behavior that minimizes its cost (breaking ties randomly), given the current behavior of its neighbors.

The theorem below assumes that the underlying network is $\Delta$-regular. 
However, as in the case of Theorem \ref{thm:main}, the theorem can be generalized to arbitrary networks with minimal degree $\Delta$ (SI, Theorem \ref{thm:generalized_version}).
\begin{theorem}\label{thm:second}
Consider a $\Delta$-regular network $G$ with $n$ players undergoing the two-order model.  Assume that the following two conditions hold.
\begin{itemize}
    \item {\em Condition $(i)$} $\alpha_2 < \beta_2$, and
    \item {\em Condition $(ii)$} $\alpha_1 < \Delta \beta_1$.
\end{itemize}
Then, with probability at least $1-\frac{1}{c^n}$, for some constant $c>1$, in at most $3 \cdot \diam(G)+1$ rounds, the system will be in a configuration in which all players are cooperative, and will remain in this configuration forever.
\end{theorem}

\section*{Discussion}

This paper proposes a simple idealized network model that demonstrates how cooperation can emerge, even when the MPCR is zero, and even when the extent of social-pressure is low. Our results highlight the possible social role that might be played by hypocritical behavior in escaping the tragedy-of-the-commons. The main finding is that setting the level of social-pressure towards this behavior to be at a specific intermediate range allows to quickly transform an almost completely defective system into a fully cooperative one. Our model, like any model, neglects many of the real-life complexity parameters. Nevertheless, the insight we discovered sheds new light on the possibility of emergent cooperation. In particular, our results suggest that individuals who wish to influence others in the context of environmental preservation should rethink their relation to their hypocritical acquaintants.

\section*{Methods}

For two players $u$ and $v$ in $G$, let $d_G(u,v)$ denote the {\em distance} between $u$ and $v$, that is, the number edges on the shortest path linking $u$ to $v$ in $G$. The maximal distance between any pair of players, i.e., the {\em diameter}, is denoted by
$ \diam(G) = \max_{u,v \in G} d_G(u,v)$. 

The initial configuration is governed by a given fixed   $0<\epsilon<1$, which is independent from the underlying graph. In the main model, each player is initially set to be a defector with probability $1-\epsilon$, a hypocritical with probability $\epsilon/2$, and a cooperative with probability $\epsilon/2$. 
Similarly, in the two-order model, each player is initially chosen to be a defector, with probability $1-\epsilon$, and, otherwise, with probability $\epsilon$ it chooses one of the three remaining behaviors with equal probability, i.e., $\epsilon/3$. We think of $\epsilon$ as very small; for example, in each of our simulations we take $\epsilon=0.01$, which means that initially,  $99\%$ of the population are defectors, $0.5\%$ are hypocritical, and $0.5\%$ are cooperators. 

We simulated the dynamics of the main model
using the C++ language. Figured were obtained using the Python library ``Matplotlib''.
In Figures~\ref{fig:grid_evolution} and \ref{fig:grid_eh_rhoh} we used a $50 \times 50$, $4$-regular, torus grid.
In Figures~\ref{fig:regular_evolution} and \ref{fig:regular_eh_rhoh} we used random 10-regular networks with $1000$ vertices. To sample such a network, we gradually increased the number of edges, by pairing the vertices of degree less than $10$ uniformly at random, until it became not possible anymore; then we discarded the few ``left-overs'' if necessary. As a consequence, the sampled networks have sometimes slightly less than $1000$ vertices, but are always 10-regular by construction.

When running the time-simulations on the grid in Figure~\ref{fig:grid_evolution}, we took $E_h = 0.1$, $\rho_d = 0.45$, and $\rho_h = 0.23$.
In Figure~\ref{fig:regular_evolution}, the time-simulation was executed on a single random $10$-regular network, using the parameters $E_h = 0.1$, $\rho_d = 0.22$, and $\rho_h = 0.11$. For both cases  these parameters  satisfy the constraints in~Eq.~\eqref{eq1}.
The insets show the evolution of the population when hypocritical behavior is disabled. This means that each player must choose between cooperating and defecting only, and that in the initial configuration, each player is a defector with probability $1-\epsilon$, and a cooperator with probability $\epsilon$. The setting remains otherwise unchanged.

In both Figure~\ref{fig:grid_eh_rhoh} and Figure~\ref{fig:regular_eh_rhoh}, the results of the simulations are presented for 150 values of $E_h$ and 150 values of $\rho_h$, with $E_h \in [ 0,1 ]$ and $\rho_h \in [0,\rho_d]$. For each couple $(E_h,\rho_h)$, a pixel is drawn at the appropriate location, whose RGB color code corresponds to the proportions of defectors (red), cooperators (green), and hypocritical players (blue) in steady-state -- that is, after T rounds. These proportions have been averaged over N repetitions, with each time a new starting configuration, and, in the case of 10-regular networks (Figure~\ref{fig:regular_eh_rhoh}), a new network. For the grid, we set $T=20, N=10$, whereas for the 10-regular networks, we took $T=10, N=100$. 

Figure~\ref{fig:grid_eh_rhoh_EG} and Figure~\ref{fig:regular_eh_rhoh_EG} are obtained in the same way as Figures~\ref{fig:grid_eh_rhoh} and~\ref{fig:regular_eh_rhoh}, respectively, except that players do not choose greedily their behavior for the next round. Instead, at each round, each player chooses a behavior that minimizes its cost (breaking ties randomly) with probability $0.95$, and otherwise chooses a behavior uniformly at random.

\paragraph{} All the experiments mentioned in this paper are numerical simulations. Specifically, they do not involve any real participant.

\subsection*{Acknowledgments} The authors would like to thank Yannick Viossat, Pierre Fraigniaud, and Ofer Feinerman for helpful discussions. This work has received funding from the European Research Council (ERC) under the European Union's Horizon 2020 research and innovation program (grant  agreement No 648032).

This is a preprint of an article published in Scientific Reports. The final authenticated version is available online at:  \url{https://doi.org/10.1038/s41598-021-97001-3}.

% Bibliography
\bibliographystyle{unsrt}
\bibliography{ref}

%=======================================================================
\clearpage
\setcounter{page}{1}
%\bigskip
\appendix
\centerline{\Huge Supplementary Information}
%=======================================================================

\section{Preliminaries} \label{secSI:preliminary}
\subsection{Definitions}
Let $G$ be a connected, undirected network. In the context of our models, we often refer to the vertices of $G$ as players. Given a player~$u$, we write $N(u)$ the set of {\em neighbors} of~$u$. Similarly, given a subset~$A$ of players, we write~$N(A)$ the set of neighbors of~$A$, that is
\[N(A) = \bigcup_{u \in A} N(u). \]
The {\em degree} of a player $u$ is the number of its neighbors, that is $|N(u)|$. We say that $G$ has minimal degree $\Delta$ if every player has degree at least $\Delta$. 
A network is called {\em $\Delta$-regular} if all the vertices have degree precisely~$\Delta$.

For two players $u$ and $v$ in $G$, let $d_G(u,v)$ denote the {\em distance} between $u$ and $v$, that is, the number edges on the shortest path linking $u$ to $v$ in $G$. Similarly, given a subset~$A$ of players, we write~$d_G(u,A)$ the distance between $u$ and $A$, that is
$$ d_G(u,A) = \min_{v \in A} d_G(u,v). $$
The {\em diameter} of $G$, is
$$ \diam(G) = \max_{u,v \in G} d_G(u,v).$$

 A {\em bipartite network} is a network $G$ whose set of vertices can be divided into two disjoint  sets $U$ and $V$, such that every edge connects a player in $U$ to a player in $V$. It is a well-known fact that a network is bipartite network if and only if it does not contain any odd-length cycles \cite{diestel2005graph}. 

\subsection{A result in graph theory}
The following lemma (mentioned also in \cite{mohammadian2016generalization}) appears to be a basic result  in graph theory, however, we could not find a formal proof for it. We therefore provide a proof here for the sake of completeness.
\begin{lemma} \label{lem:odd-girth}
    The shortest odd-length cycle of any non-bipartite network~$G$ is of length at most $2\diam(G)+1$.
\end{lemma}
\begin{proof}
Consider a non-bipartite network~$G$. Such a network necessarily has an odd-length cycle. Let $2k+1$ be the shortest length among the odd-length cycles in $G$, where $k$ is an integer, and let $C = (u_1,\ldots,u_{2k+1})$ be such a cycle.
    
    \begin{claim} \label{claim:graph_aux1}
        For every $i,j \in \{ 1,\ldots, 2k+1 \}$ such that $d_C(u_i,u_j) \geq 2$, there exist $\ell \neq i,j$ and a shortest path $P$ between $u_i$ and $u_j$ such that $P$ contains $u_\ell$. 
    \end{claim}
    \begin{proof} [Proof of Claim~\ref{claim:graph_aux1}]
        Assume by way of contradiction that we can find $i < j$ such that no shortest path between $u_i$ and $u_j$ has any intermediate vertex among $\{ u_1,\ldots,u_{2k+1} \}$. Up to re-indexing the vertices of the cycle, we can assume that $j-i \leq k$. Let $(u_i = v_1,v_2,\ldots,v_s,v_{s+1}=u_j)$ be a shortest path between $u_i$ and $u_j$. By assumption, $\{ v_2,\ldots,v_s \} \cap \{ u_1,\ldots, u_{2k+1} \} = \emptyset$, and $s < j-i$ (otherwise $(u_i,u_{i+1},\ldots,u_{j-1},u_j)$ is a shortest path). Consider two cases:
\begin{itemize}
    \item 
        If $s$ and $j-i$ have different parities, then $s+j-i$ is odd. Moreover, $s+j-i \leq 2(j-i) \leq 2k$, so
        $$ (v_1 = u_i,u_{i+1},\ldots,u_{j-1},u_j = v_{s+1},v_s, \ldots,v_2) $$
        is an odd-length cycle shorter than $C$, which is a contradiction.
      \item  
        If $s$ and $j-i$ have the same parity, then $2k+1 + s-(j-i)$ is odd. Moreover, $2k+1 + s-(j-i) < 2k+1$, so
        $$ (u_i = v_1,v_2,\ldots,v_s,v_{s+1} = u_j,u_{j+1}, \ldots,u_{2k+1},u_1, \ldots,u_{i-1}) $$
        is again an odd-length cycle shorter than $C$, which is a contradiction. \end{itemize} 
        This concludes the proof of Claim~\ref{claim:graph_aux1}.
    \end{proof}
    
    \begin{claim} \label{claim:graph_aux2}
        For every $i,j \in \{ 1,\ldots, 2k+1 \}$, there exist a shortest path $P$ between $u_i$ and $u_j$ such that $P$ contains only vertices of $C$ -- in other words, $d_C(u_i,u_j) = d_G(u_i,u_j)$.
    \end{claim}
    \begin{proof} [Proof of Claim~\ref{claim:graph_aux2}]
        We prove the claim by induction on $d_C(u_i,u_j)$, the distance between $u_i$ and $u_j$ in $C$. When $d_C(u_i,u_j) = 1$, $(u_i,u_j)$ is a path of length $1$ between $u_i$ and $u_j$ containing only vertices of $C$. Next, let us assume that the claim holds for every pair of vertices whose distance in $C$ is at most $1\leq d\leq k$. Consider $i$ and $j$ such that $d_C(u_i,u_j) = d+1$. By Claim~\ref{claim:graph_aux1}, we can find $\ell$ and a shortest path $P$ between $u_i$ and $u_j$ such that $P$ contains $u_\ell$. By the induction hypothesis, we can find shortest paths $P_1$ between $u_i$ and $u_\ell$, and $P_2$  between $u_\ell$ and $u_j$, such that $P_1$ and $P_2$ contain only vertices of $C$. By merging $P_1$ and $P_2$, we obtain a shortest path between $u_i$ and $u_j$ containing only vertices of $C$, which establishes the induction step. This concludes the proof of Claim~\ref{claim:graph_aux2}.
    \end{proof}
    By Claim~\ref{claim:graph_aux2}, $k = d_C(u_1,u_{k+1}) = d_G(u_1,u_{k+1}) \leq \diam(G)$, where the last inequality is by the definition of  diameter. Hence, $2k+1 \leq 2\diam(G)+1$.  This concludes the proof of Lemma \ref{lem:odd-girth}.
    \end{proof}
\section{Proof of Theorem \ref{thm:main}} \label{secSI:proofs}
The goal of this section is to prove Theorem \ref{thm:main}. In fact, we prove the more general theorem below. 
\begin{theorem}\label{thm:main-SI}
Consider a network $G$ with $n$ players and minimal degree $\Delta$. Assume that the following conditions hold. 
\begin{itemize}
    \item {\em Condition $(i)$} $E_h+\rho_{h}<\rho_{d}$, and
    \item {\em Condition $(ii)$} $E_h+\rho_{h}\cdot\Delta>1$.
\end{itemize}
Then, for some constant $c>1$ (that depends only on $\epsilon$ and not on~$G$) the following holds.
\begin{itemize}
    \item 
If $G$ is not bipartite then  with probability at least $1-\frac{1}{c^n}$, in at most $3 \cdot \diam(G) + 1$ rounds, the system will be in a configuration in which all players are cooperative, and will remain in this configuration forever.  
\item 
If $G$ is bipartite and $\Delta$-regular  then with probability at least $1-\frac{1}{c^n}$, in at most $\diam(G) + 1$ rounds, the system will be in a configuration in which all players are cooperative, and will remain in this configuration forever. 
\item
If $G$ is bipartite then with probability at least $1-\frac{1}{c^\Delta}$, in at most $\diam(G) + 1$ rounds, the system will be in a configuration in which all players are cooperative, and will remain in this configuration forever.
\end{itemize}
\end{theorem}
Before we prove Theorem \ref{thm:main-SI} we note that in the third item, the probability bound of $1-\frac{1}{c^\Delta}$ is tight for bipartite graphs, up to replacing $c$ with another constant. Indeed, consider the bipartite graph which is constructed by having $\Delta$ players in $U$, each of which is connected to each of the remaining $n-\Delta$ players in $V$. Then, with probability  $\frac{1}{c^\Delta}$, for some constant $c$, all players in $U$ are defectors initially. In this case, it is possible to show that, regardless of the relationships between $\rho_d$, $\rho_h$ and $E_h$, but as long as being a defector is the best choice when all neighbors are defectors, the system keeps alternating forever, so that on even rounds all players in $U$ are defectors, and on odd rounds all players in $V$ are defectors.

\begin{proof}[Proof of Theorem \ref{thm:main-SI}]
    We start with defining  $\D_t$ as the set of non-defector players at round~$t$.
    The following lemma describes the propagation of the non-defector state in the network.
    It says that a player $u$ is a non-defector at round $t+1$ if and only if at least one of its neighbors $v$ is a non-defector in round $t$. Note, however, that this does not imply that the neighbor $v$ remains a non-defector in the next round as well.

\begin{lemma} \label{lem:contagion}
    Under Condition {\em (i)}, $\D_{t+1} = N(\D_t)$.
\end{lemma}
\begin{proof}
    First, we prove that $N(\D_t) \subseteq \D_{t+1}$.
    Let $u \in N(\D_t)$. By definition, there exists a neighbor $v$ of $u$ such that $v$ is a non-defector at round $t$. We claim that for $u$, being a hypocritical in round $t+1$ is strictly more beneficial than being a defector. Indeed, as a hypocritical it will pay $E_h+ \rho_{h} \cdot \DeltaD(u)$, and as a defector it will pay 
    $\rho_{d} \cdot \DeltaD(u)$. Since $v$ is non-defector then $\DeltaD(u) \geq 1$, and hence:
    \[E_h+ \rho_{h} \cdot \DeltaD(u)\leq (E_h+ \rho_{h}) \cdot \DeltaD(u)< \rho_{d} \cdot \DeltaD(u),\]
    where we used Condition {\em (i)} in the last inequality. 
    Therefore, the cost of $u$ as a defector is strictly higher than its cost as a 
    hypocritical. This implies that in the next round $u$ will be either a hypocritical or a cooperative player, i.e., $u \in \D_{t+1}$.
    
    To prove the other inclusion, $\D_{t+1} \subseteq N(\D_t)$, consider a player~$u \notin N(\D_t)$, i.e., having only defectors as neighbors at round $t$, or in other words, at round $t$, we have $\DeltaD(u) = 0$.  If~$u$ chooses to be a defector at round $t+1$, then it would pay~$\DeltaD(u) \rho_d = 0$, which is less than what it would pay as a hypocritical ($E_h + \DeltaD(u) \rho_h = E_h$) or cooperator ($1$). Hence, $u \notin \D_{t+1}$.
\end{proof}

\begin{lemma} \label{lem:second_stage}
    Assume that Conditions {\em (i)} and {\em (ii)} hold, and assume that for some round~$t_0$ all players are non-defectors. Then, at round~$t_0+1$, all players will be cooperative, and will remain cooperative forever.
\end{lemma}
\begin{proof}
Assume that at round~$t_0$ all players are non-defectors. 
    By Lemma~\ref{lem:contagion}, we know that every player will remain non-defector for every round after~$t_0$. It therefore remains to show, that at the end of round~$t$, for any $t\geq t_0$, being a cooperative is strictly more beneficial than being a hypocritical. 
    
    Observe that since each player has at least $\Delta$ neighbors, and since all neighbors are non-defectors at round~$t$, then for every player~$u$, we have $\DeltaD(u) \geq \Delta$ at round $t$.
    Therefore, being a hypocritical costs $E_h+\rho_{h} \cdot \DeltaD(u) \geq E_h+\rho_{h} \cdot \Delta$.
    By Condition {\em (ii)}, this quantity is strictly greater than 1, hence more than what a cooperative player would pay.
  It follows that, at the end of round $t$, being a cooperative is strictly more beneficial than being a hypocritical, implying that all players would be cooperators at round $t+1$. This completes the proof of Lemma \ref{lem:second_stage}.
\end{proof}

\begin{lemma} \label{lem:punishing_pair}
    Assume that Conditions {\em (i)} and {\em (ii)} hold, and assume that for some round $t_0$, we have $\D_{t_0} \cap N(\D_{t_0}) \neq \emptyset$, that is, there are at least two neighboring non-defectors. Then in at most $\diam(G) + 1$ rounds as of round $t_0$,
    the system will be in the configuration in which all players are cooperative, and will remain in this configuration forever.
\end{lemma}
\begin{proof}
    By assumption, there exists two neighbors $u_0, u_0' \in \D_{t_0}$. We define inductively a sequence of sets $\{U_j\}_{j\geq 0}$, setting $U_0 = \{ u_0, u_0' \}$, and for every   $j$, defining $U_{j+1} = N(U_j)$.
    
    \begin{claim} \label{claim:induction_U_j}
        For every integer~$j\geq 0$,
        \begin{equation} \label{eq:claim1}
            U_j \subseteq N(U_j)
        \end{equation}
        (each player in $U_j$ has at least one neighbor in $U_j$), and
        \begin{equation} \label{eq:claim2}
            U_j \subseteq \D_{t_0 + j} 
        \end{equation}
        (each player in $U_j$ is non-defector at round $t_0 + j$).
    \end{claim}
    \noindent{\em Proof of Claim \ref{claim:induction_U_j}.} The proof proceeds by induction. 
        The base of the induction, corresponding to $j=0$, is true by the assumption on $u_0$ and $u_0'$. Next, let us assume that the claim holds for some integer $j\geq 0$. By the induction hypothesis with respect to~Eq.~\eqref{eq:claim1}, $U_j \subseteq N(U_j)$, so $N(U_j) \subseteq N(N(U_j))$, and hence, by definition, $U_{j+1} \subseteq N(U_{j+1})$.
        In other words, we have proved that~Eq.~\eqref{eq:claim1} holds at round $j+1$. Next, by the induction hypothesis with respect to~Eq.~\eqref{eq:claim2}, we have $U_j \subseteq \D_{t_0 + j}$, so $N(U_j) \subseteq N(\D_{t_0 + j})$. By definition of $U_{j+1}$, and by Lemma~\ref{lem:contagion}, we can rewrite this as $U_{j+1} \subseteq \D_{t_0 + j + 1}$, establishing~Eq.~\eqref{eq:claim2} at round $j+1$. This completes the induction step and concludes the proof of Claim \ref{claim:induction_U_j}. \qed\\

    A direct consequence of~Eq.~\eqref{eq:claim1} in Claim \ref{claim:induction_U_j} and the definition of the sequence $\{U_j\}_j$ is that 
    $U_{j+1} = U_j \cup N(U_j)$, and so, $U_{j+1}$ is equal to $U_j$ together with all the neighbors of players in $U_j$. As a consequence, for every $j \geq \diam(G)$, the set $U_j$ contains all players. By~Eq.~\eqref{eq:claim2} of Claim~\ref{claim:induction_U_j}, this implies that from round $t_0 + \diam(G)$ onward, all players are non-defectors.
    
    By Lemma~\ref{lem:second_stage}, we conclude that from round~$t_0 + \diam(G) + 1$ onward, all players are cooperative. This completes the proof of Lemma \ref{lem:punishing_pair}.
\end{proof}

\begin{lemma} \label{lem:non-bipartite}
    Assume that Conditions {\em (i)} and {\em (ii)} hold, and that $G$ is not bipartite. If $\D_0 \neq \emptyset$, i.e., if initially there is at least one non-defector player, then in at most $3 \cdot \diam(G) + 1$ rounds,
    the system will be in the configuration in which all players are cooperative, and will remain in this configuration forever.
\end{lemma}
\begin{proof}
    By assumption, $G$ is not bipartite, or equivalently, $G$ contains at least one odd-length cycle. Let $(u_1,\ldots,u_{2k+1})$ be a shortest odd-length cycle of~$G$. Given $s = d_G(u_1,\D_0)$, let $(v_0 \in \D_0, v_1,\ldots,v_{s-1},v_s = u_1)$ be a shortest path from $\D_0$ to $u_1$. By Lemma~\ref{lem:contagion}, it follows by induction that for every $t \in \{0,\ldots,s\}$, $v_t \in \D_t$, and hence, $u_1 \in \D_s$ (note that, although $v_{t-1} \in \D_{t-1}$, it could be that $v_{t-1} \notin \D_t$). Similarly, for every $t \in \{1,\ldots,k\}$, $u_{1+t} \in \D_{s+t}$ and $u_{2k+2-t} \in \D_{s+t}$. Hence, $u_{k+1} \in \D_{s+k}$ and $u_{k+2} \in \D_{s+k}$. In other words, we have just showed that in round ${s+k}$, we have two non-defector neighbors. 

    By the definition of diameter, $s \leq \diam(G)$. By Lemma~\ref{lem:odd-girth}, we also have~$k \leq \diam(G)$. By Lemma~\ref{lem:punishing_pair}, the system needs at most $\diam(G)+1$ rounds after round $s+k$ to reach full cooperation. We conclude that it reaches cooperation in at most $3\cdot\diam(G)+1$ rounds, as stated.
\end{proof}

\begin{lemma} \label{lem:bipartite}
    Assume that Conditions {\em (i)} and {\em (ii)} hold, and that $G$ is bipartite. The set of players can be partitioned into $U$ and $V$ such that $U \cap N(U) = V \cap N(V) = \emptyset$. If $\D_0 \cap U \neq \emptyset$ and $\D_0 \cap V \neq \emptyset$, then in at most $T=\diam(G) + 1$ rounds, 
    the system will be in the configuration in which all players are cooperative, and will remain in this configuration forever.
\end{lemma}
\begin{proof}
    By assumption, $G$ is bipartite.
    We define inductively a sequence of subsets of the set of players, $U_0 = U \cap \D_0$, and for every $k \geq 0$, $U_{k+1} = N(N(U_t))$ -- that is, $U_{k+1}$ contains the neighbors (in $U$) of the neighbors (in $V$) of the players in $U_k$. Note that, as a consequence of this definition, $U_k \subseteq U_{k+1}$. Let $k_0 = \floor{\diam(G) / 2}$.
    
    Let us show that $U_{k_0} = U$ (and hence, that for every $k \geq k_0$, $U_k = U$). For this purpose, consider a player $u \in U$. Let $(u_0,v_0,u_1,v_1,\ldots,u_{s-1},v_{s-1},u_s)$, where $u_0 \in U_0$ and $u_s = u$ be a shortest path from~$U_0$ to~$u$. This path is of length $2s \leq \diam(G)$, so $s \leq k_0$.
    Since $u_{\ell+1} \in N(N(u_\ell))$ for every $\ell \leq s$, it follows by induction on $\ell$ that for every $\ell \leq s$, $u_\ell \in U_\ell$, and hence that $u \in U_s$.
    As we have seen, the sequence $\{U_k\}_k$ is non-decreasing and since $s \leq k_0$, we obtain $u \in U_{k_0}$. This establishes that $U_{k_0} = U$.
    
    Next, we prove by induction that for every $k$, $U_k \subseteq \D_{2k}$. This is true for $k=0$ by definition. Assume that this is true for some integer~$k\geq 0$. We have
    $$ U_{k+1} = N(N(U_k)) \subseteq N(N(\D_{2k})) = \D_{2k + 2},$$
    where the second transition is by the induction hypothesis and the last transition is due to Lemma \ref{lem:contagion}.
    This concludes the induction proof. Since we have already proved that $U_{k_0} = U$, we conclude that $U \subseteq \D_{2k_0}$.
    
    We can apply the same reasoning to $V$ and obtain that $V \subseteq \D_{2k_0}$. Thus, $\D_{2k_0}$ contains all players. By Lemma~\ref{lem:second_stage}, from round~$2k_0+1 \leq 2\diam(G) + 1$ onward, the system will be in a configuration in which all players are cooperative and will remain in this configuration forever. This concludes the proof of Lemma~\ref{lem:bipartite}.
\end{proof}

Finally, we wrap the aforementioned lemmas to prove the theorem with respect to different networks. 
Recall that initially, each player is set to be a defector with probability $1-\epsilon$, a hypocritical with probability $\epsilon/2$, and a cooperative with probability $\epsilon/2$, for some fixed $0<\epsilon<1$ independent of $n$. We consider three families of networks.
\begin{itemize}
    \item If $G$ is not bipartite, then Lemma \ref{lem:non-bipartite} guarantees that the system converges to full cooperation in $3\cdot\diam(G)+1$ rounds, provided that the initial configuration contains at least one non-defector. This happens with overwhelmingly high probability, specifically,  $1-(1-\epsilon)^n=
    1-\frac{1}{c^n}$, for some constant $c>1$. This completes the proof of the first item in  Theorem \ref{thm:main-SI}.
    \item If $G$ is bipartite, then the set of players in $G$ can be split into two disjoint sets $U$ and $V$  such that all edges are between $U$ and $V$. Lemma~\ref{lem:bipartite} guarantees that the system converges to full cooperation $\diam(G)+1$ rounds, provided that there is at least one non-defector in $U$ and at least one non-defector in $V$. Let us see what is the probability that the initial configuration satisfies this.
    \begin{itemize}
        \item If $G$ is $\Delta$-regular, then both $U$ and $V$ contain precisely $n/2$ players. This follows from the fact that the number of edges outgoing from $U$, respectively $V$, is precisely $\Delta|U|$, respectively $\Delta|V|$, and these numbers are equal. In this case the probability that there is at least one non-defector in $U$ and at least one non-defector in $V$ is $\pa{1-(1-\epsilon)^{n/2}}^2 \geq 1-2(1-\epsilon)^{n/2} >1-\frac{1}{c^n}$, for some constant $c>1$. This completes the proof of the second item in  Theorem \ref{thm:main-SI}.
        
        \item For general bipartite $G$  with minimal degree $\Delta$, we have that both $|U|$ and $|V|$ are greater or equal to $\Delta$. Hence, the probability that there is at least one non-defector in $U$ and at least one non-defector in $V$ is at least $\pa{1-(1-\epsilon)^{\Delta}}^2 \geq 1-2(1-\epsilon)^{\Delta} >1-\frac{1}{c^\Delta}$, for some constant $c>1$. This completes the proof of the third item in  Theorem \ref{thm:main-SI}.
    \end{itemize}
\end{itemize}
\end{proof}

\section{Proof of Theorem \ref{thm:second}} \label{secSI:generalized_model}

The goal of this section is to prove Theorem \ref{thm:second}. In fact, we prove the more general theorem below.

\begin{theorem}\label{thm:generalized_version}
Consider a network $G$ with $n$ players and minimal degree $\Delta$ undergoing the two-order model, so that the following conditions hold.
\begin{itemize}
    \item {\em Condition $(i)$} $\alpha_2 < \beta_2$, and
    \item {\em Condition $(ii)$} $\alpha_1 < \Delta \beta_1$.
\end{itemize}
Then the following holds for some constant $c>1$.
\begin{itemize}
    \item 
If $G$ is not bipartite then  with probability at least $1-\frac{1}{c^n}$, in at most $3 \cdot \diam(G) + 1$ rounds, the system will be in a configuration in which all players are cooperative, and will remain in this configuration forever.  
\item 
If $G$ is bipartite and $\Delta$-regular  then with probability at least $1-\frac{1}{c^n}$, in at most $\diam(G) + 1$ rounds, the system will be in a configuration in which all players are cooperative, and will remain in this configuration forever. 
\item
If $G$ is bipartite then with probability at least $1-\frac{1}{c^\Delta}$, in at most $\diam(G) + 1$ rounds, the system will be in a configuration in which all players are cooperative, and will remain in this configuration forever.
\end{itemize}
\end{theorem}
Before we prove the theorem, we note that Condition $(ii)$ is necessary for the emergence of cooperation on $\Delta$-regular graphs, since having $\alpha_1 > \Delta \beta_1$
would imply that it is always beneficial to defect on the first level.
\begin{proof}[Proof of Theorem \ref{thm:generalized_version}]
Consider a network $G$, and parameters $\alpha_1$, $\alpha_2$, $\beta_1$, and $\beta_2$, satisfying Conditions  $(i)$ and $(ii)$ in Theorem~\ref{thm:generalized_version}. We first observe that under Condition {\em (i)}, at any round $t\geq 1$, no player ever chooses to be a private cooperator in $G$. Indeed, if $\Delta_2(u) \geq 1$, then a private cooperator would pay $\alpha_1  + \Delta_2(u) \spp > \alpha_1 + \Delta_2(u) \alpha_2 \geq \alpha_1 + \alpha_2$, hence, more than the cost of cooperating, and when $\Delta_2(u) = 0$, a private cooperator would pay~$\alpha_1 > 0$, hence, more than the cost of defecting. It follows that, although the initial configuration may include private cooperators, this behavior completely disappears from the system after the first round. 
    
    Next, we aim to prove Theorem~\ref{thm:generalized_version} by reducing it to Theorem~\ref{thm:main-SI}. 
    Let $G'$ be a network identical to $G$, undergoing the main model (for which Theorem~\ref{thm:main-SI} applies), taking the parameters:
    \begin{equation}\label{eq:mapped_parameters}
     E_h = \frac{\alpha_2}{\alpha_1 + \alpha_2}, ~~~\rho_h = \frac{\beta_1}{\alpha_1 + \alpha_2}, ~~~\text{ and }~~~\rho_d = \frac{\beta_1 + \beta_2}{\alpha_1 + \alpha_2}. 
       \end{equation}
    A {\em configuration} $\cal{C}$ on $G$ is an assignment of behaviors, namely, either defectors, cooperators, hypocritical, or private cooperators, to the players in $G$. Recall that the initial configuration on $G$ is sampled according to the distribution $\psi(\epsilon)$, so that each player is initially chosen to be a defector with probability $1-\epsilon$, and any of the three remaining behaviors with probability $\epsilon/3$. 
     
     We next define a mapping $f$, transforming each initial configuration $\cal{C}$ on $G$ to an initial configuration $\cal{C}'$ on $G'$. The mapping is very simple: All players in $G'$ remain with the same behavior as in $G$ except that private cooperators are turned into defectors. It is easy to see that given  the distribution $\psi(\epsilon)$, this mapping induces the distribution $\psi'(\epsilon')$ 
        on the initial configurations in $G'$, where $\epsilon'=\frac{2}{3}\epsilon$. Indeed, under this mapping, a player in $G'$ is initially chosen to be a defector with probability $1-\epsilon+\epsilon/3=1-\epsilon'$, a cooperator with probability $\epsilon/3=\epsilon'/2$, and hypocritical with probability $\epsilon/3=\epsilon'/2$.
     
     At this point, we address a technicality that concerns the randomness involved in breaking ties. That is, recall that at any round $t$, if the minimal cost is attained by several behaviors then the player chooses one of them uniformly at random. One way to implement this is by considering 
     a certain order between the behaviors, and sampling a number uniformly at random  $r\in [0,1]$. For instance, consider the following ordering:  $\mbox{cooperator} > \mbox{hypocritical} > \mbox{defector}$ (as we saw, in the regime of parameters we consider, a private cooperator in the two-order model never attains the minimal cost, and hence it is never considered as an option). 
     If a player needs to choose, say, between being a cooperator or a defector, then it chooses to be a cooperator if $r$ is in $[0,0.5]$, and otherwise, it chooses to be a defector. This means that given a sequence of random numbers $\{r_i\}_{i=1}^\infty$, where $r_i\in [0,1]$, the behavior of a player is deterministically described by the behaviors of its neighbors at each round. 

Consider a fixed sequence of random numbers $\{r_i\}_{i=1}^\infty$, where $r_i\in [0,1]$.
     Let ${\cal{C}}_0$ be an initial configuration in $G$, and let ${\cal{C}}_t$ denote the configuration ${\cal{C}}_0$ at round $t$, with the costs defined according to the two-order model on $G$, and using the sequence $\{r_i\}_{i=1}^\infty$ to break ties if necessary.
     Let ${\cal{C}}'_0=f({\cal{C}}_0)$ be the mapped configuration on $G'$, 
     and let ${\cal{C}}_t'$ 
     be the corresponding configuration at round $t\geq 1$, with the costs defined according to the parameters mentioned in~Eq.~\eqref{eq:mapped_parameters}, and using the same sequence $\{r_i\}_{i=1}^\infty$ to break ties if necessary. 
     
     \begin{claim}\label{claim:same_behavior}
     For every $t\geq 1$, we have 
        \[
        {\cal{C}}_t'={\cal{C}}_t,
        \]
    \end{claim}
    \noindent{\em Proof of Claim \ref{claim:same_behavior}.} 
    Our goal is to show that at any round~$t\geq 0$, a player~$u$ in ${\cal{C}}_t$ is a defector, a cooperator or a hypocritical, respectively, if and only if it is a defector, a cooperator or a hypocritical, respectively, in ${\cal{C}}'_t$, and that a private cooperator in ${\cal{C}}_t$ is a defector in ${\cal{C}}'_t$.
    
    Let us prove this claim by induction. By definition, the claim holds for $t=0$. Assume that it holds for some integer~$t\geq 0$. By the induction hypothesis, for every player $u$, the set $\Delta_2(u)$ in $G$ is equal to $\DeltaD(u)$ in $G'$. Hence, with our definitions of $E_h, \rho_h$ and $\rho_d$ in~Eq.~\eqref{eq:mapped_parameters}, we argue that the cost of being a cooperator in ${\cal{C}}'_t$ (respectively hypocritical, defector) is $\frac{1}{\alpha_1 + \alpha_2}$ times the cost of being a cooperator in ${\cal{C}}_t$ (respectively hypocritical, defector).
    Indeed, a cooperator in ${\cal{C}}'_t$ pays
    $$ 1 = \frac{1}{\alpha_1 + \alpha_2} \cdot (\alpha_1 + \alpha_2), $$
    while $(\alpha_1 + \alpha_2)$ is what it pays in ${\cal{C}}_t$. A hypocritical player in ${\cal{C}}'_t$ pays
    $$ E_h + \DeltaD(u) \cdot \rho_h = \frac{\alpha_2}{\alpha_1 + \alpha_2} + \Delta_2(u) \frac{\beta_1}{\alpha_1 + \alpha_2}, $$
    while $(\alpha_2 + \Delta_2(u) \cdot \beta_1)$ is what a hypocritical pays in ${\cal{C}}_t$, and 
    a defector in ${\cal{C}}'_t$ pays 
    $$ \DeltaD(u) \cdot\rho_d = \Delta_2(u) \cdot \frac{\beta_1 + \beta_2}{\alpha_1 + \alpha_2}, $$
    while $\Delta_2(u) \cdot (\beta_1 + \beta_2)$ is what it pays in ${\cal{C}}_t$.
    
    Moreover, recall that no player in $G$ ever chooses to be a private cooperator in rounds $t\geq 1$. 
    Hence, the behavior that minimizes the cost in $G$ is the same as in $G'$.
    It follows that at round $t+1$, all players choose the same behavior in $G$ as they would in $G'$, which establishes the induction proof, and concludes the proof of Claim \ref{claim:same_behavior}.
    \qed

    Next, we prove that with our choices of $E_h, \rho_h$ and $\rho_d$ in~Eq.~\eqref{eq:mapped_parameters}, Conditions {\em (i)} and {\em (ii)} in Theorem~\ref{thm:generalized_version} imply Conditions {\em (i)} and {\em (ii)} in Theorem~\ref{thm:main-SI}:
    \begin{align*}
        \alpha_2 < \beta_2 &\iff \alpha_2 + \beta_1 < \beta_1 + \beta_2 \\
        &\iff \frac{\alpha_2}{\alpha_1 + \alpha_2} + \frac{\beta_1}{\alpha_1 + \alpha_2} < \frac{\beta_1 + \beta_2}{\alpha_1 + \alpha_2} \\
        &\iff E_h+\rho_{h}<\rho_{d},
    \end{align*}
    and
    \begin{align*}
        \Delta \cdot \beta_1 > \alpha_1 &\iff \alpha_2 + \Delta \cdot \beta_1 > \alpha_1 + \alpha_2 \\ 
        &\iff \frac{\alpha_2}{\alpha_1 + \alpha_2} + \Delta \cdot \frac{\beta_1}{\alpha_1 + \alpha_2} > 1 \\
        &\iff E_h+\rho_{h}\cdot\Delta>1.
    \end{align*}
    Hence, we can apply Theorem~\ref{thm:main-SI} to the mapped process on $G'$. 
    It follows that in the number of rounds and probability guarantees as stated in Theorem~\ref{thm:main-SI}, $G'$ converges to the configuration in which all players are cooperators, and remains in that configuration forever. 
    By Claim \ref{claim:same_behavior}, this holds for $G$ as well, concluding the proof of Theorem \ref{thm:generalized_version}. 
\end{proof}

\end{document}